\documentclass[12pt]{article}

\usepackage{float}
\usepackage[makeroom]{cancel}
\usepackage{color}
\usepackage{graphicx}
\usepackage{amsmath}
\usepackage{amssymb}
\usepackage{xspace}
\usepackage[small]{subfigure}
\usepackage[numbers,compress]{natbib}
\usepackage[hyperfootnotes=false]{hyperref}
\usepackage{tcolorbox}

\usepackage{framed}
\usepackage{physics}
\usepackage{tensor}

\newlength{\fighskip} \fighskip=2pt
\newlength{\figvskip} \figvskip=3pt

\usepackage{mciteplus} 
\usepackage{dcolumn}
\usepackage{bm}
\usepackage{verbatim}
\usepackage{amscd}
\usepackage{amsfonts}
\usepackage{setspace}
\usepackage{amsthm}
\usepackage{enumerate}
\usepackage{mathtools}

\theoremstyle{plain}
\newtheorem{theorem}{Theorem}
\newtheorem{lemma}[theorem]{Lemma}

\theoremstyle{definition}

\DeclareMathOperator{\Sep}{Sep}
\DeclareMathOperator{\supp}{supp}

\newcommand{\defeq}{\stackrel{\text{def}}{=}}

\newcommand{\ED}{E_{D}}
\newcommand{\ER}{E_{R}}
\newcommand{\EF}{E_{F}}

\newcommand{\IAB}{I(A{:}B)_{\rho}}
\newcommand{\dmin}{d_{-}}
\newcommand{\dmax}{d_{+}}

\usepackage{authblk}

\title{
\vspace{-80pt}
\hfill
{\normalsize RUP-26-18}\\
\vspace{40pt}
\bf 
Relative entropy of entanglement of \\ Haar random states
}
\author[1,2]{Takato Mori\thanks{takato.mori@yukawa.kyoto-u.ac.jp}}
\author[3,4]{Beni Yoshida\thanks{byoshida@perimeterinstitute.ca}}
\affil[1]{\em \small Department of Physics, Rikkyo University, \protect\\
3-34-1 Nishi-Ikebukuro, Toshima-ku, Tokyo 171-8501, Japan}
\affil[2]{\em \small RIKEN Quantum, Transformative Research Innovation Platform of RIKEN platforms (TRIP) Headquarters, RIKEN, Wako 351-0198, Japan}
\affil[3]{\em \small Perimeter Institute for Theoretical Physics, Waterloo, Ontario N2L 3W8, Canada}
\affil[4]{\em \small 
Fundamental Quantum Science Program (FQSP), RIKEN, Wako 351-0198, Japan}
\date{}

\begin{document}
\maketitle

\begin{abstract}
We determine the relative entropy of entanglement of a bipartite mixed state $\rho_{AB}$ obtained by tracing out one subsystem of a tripartite Haar-random pure state $|\psi\rangle_{ABC}$, finding $E_R(\rho_{AB})=\log\frac{d_Ad_B}{\max(d_A,d_B,d_C)}+O(1)$. 
Equivalently, the relative entropy of entanglement nearly saturates the smaller of the entanglement of formation $E_F(\rho_{AB})$ and the mutual information $I(A:B)$. 
The upper bound is achieved by an explicit separable state obtained through one-sided Schmidt dephasing, which is therefore approximately closest. 
\end{abstract}

\section{Introduction}

Strongly interacting quantum many-body systems are notoriously difficult to study microscopically. 
Fortunately, many questions of physical interest concern universal properties that are insensitive to detailed dynamics. 
A useful approach is therefore to replace a complicated system by an appropriate random ensemble and study typical properties.

Randomness plays a crucial role across modern physics. 
Random matrix theory captures universal spectral statistics from symmetry alone, while the eigenstate thermalization hypothesis describes thermalization through the pseudorandom structure of many-body eigenstates. 
In quantum gravity, Haar-random states and unitaries provide simple models of black-hole dynamics and underlie the black-hole evaporation mechanism. 
Haar-random tensors also form the basic building blocks of random tensor network models of holography. 
In quantum information, random states and channels similarly provide foundation for coding, scrambling, and entanglement.

Despite this broad relevance, the entanglement structure of Haar-random states is fully understood only in limited settings. 
For a bipartite Haar-random pure state, the reduced state of the smaller subsystem is nearly maximally mixed, so the state contains nearly maximal bipartite entanglement. 
The situation becomes considerably richer for a tripartite Haar-random state $\ket{\psi}_{ABC}$. After tracing out $C$, the reduced state $\rho_{AB}=\Tr_C\ket{\psi}\!\bra{\psi}$ is mixed, and different measures of entanglement need not agree.
Substantial progress has been made for several quantities, including the entanglement of formation, logarithmic negativity, and distillable entanglement under different classes of local operations~\cite{Hayden:2005sqo, Horodecki:2009zz, Aubrun:2014threshold, Lu:2020jza, Shapourian:2020mkc, Mori:2024gwe, Li:2025nxv}. 
Nevertheless, a general picture of the mixed-state entanglement of Haar-random states is still lacking.

In this work, we determine the relative entropy of entanglement $\ER(\rho_{AB})$ for Haar-induced mixed states. 
The relative entropy of entanglement quantifies the information-theoretic distinguishability from $\rho_{AB}$ to the set of separable states~\cite{Vedral:1997qn}. 
Despite its conceptual simplicity and importance in the resource theory of entanglement, $\ER$ is generally difficult to evaluate because it requires an optimization over the full set of separable states. 
Indeed, the problem is NP-hard in general~\cite{Huang_2014}, and explicit results are known only for special families of states such as pure bipartite states, Schmidt correlated states, symmetric states, graph states, and some special classes of two-qubit states~\cite{Vedral:1997hd,Audenaert:2002lkg,Vollbrecht:2001vsi,Hajdusek:2012eqs,Markham:2007vcs,Chen:2002ubp,Kim:2010fps}.

Here, we obtain a simple asymptotically exact expression for $\ER(\rho_{AB})$ for random mixed states. In particular, we find
\begin{equation}
  \ER(\rho_{AB})
  =
  \min\bigl\{
    \EF(\rho_{AB}),
    I(A{:}B)_\rho
  \bigr\}
  +O(1),
\end{equation}
showing that one of the two general upper bounds on $\ER$ is nearly
saturated throughout the parameter regime studied here. We also identify
an explicit separable state that nearly attains the optimization defining
$\ER$. This state is obtained by dephasing one subsystem in its Schmidt basis before tracing out the purifying system.
To our knowledge, the relative entropy of entanglement of generic Haar-induced mixed states has not previously been determined.

\subsection{Relative entropy of entanglement}

Let $\Sep(A{:}B)$ denote the set of states separable across $A{:}B$. 
Here, a separable state is a probabilistic mixture of product states, $\sigma_{AB}=   \sum_i p_i \sigma_A^{(i)}\otimes\sigma_B^{(i)}$.
The relative entropy of entanglement measures how far a state is from being separable:
\begin{equation}
  \ER(\rho_{AB})
  \defeq\inf_{\sigma_{AB}\in\Sep(A{:}B)}
  D(\rho_{AB}\Vert\sigma_{AB}),
  \label{eq:ER-definition}
\end{equation}
where $D(\rho\Vert\sigma) \defeq\Tr\!\left[\rho(\log\rho-\log\sigma)\right]$ is the relative entropy.
As an example, for a bipartite pure state $\ket{\phi}_{AB}$ with the Schmidt decomposition $\ket{\phi}_{AB}=\sum_i\sqrt{p_i}\ket{i}_A\otimes\ket{\varphi_i}_B$, it matches with von Neumann entropy:
\begin{equation}
  \ER(\phi_{AB})=S(\phi_A)=-\sum_i p_i\log p_i. 
  \label{eq:ER-pure-example}
\end{equation}
The closest separable state is obtained by replacing entanglement with classical correlation, namely, the classicalized state $\sum_i p_i\dyad{i}_A\otimes\dyad{\varphi_i}_B$ in the Schmidt basis. 

The relative entropy of entanglement is closely related to, but distinct from other entanglement measures. A standard hierarchy of entanglement measures gives~\cite{Devetak:2003zfw,Horodecki:2009zz}
\begin{equation}
  \mathrm{hash}(\rho_{AB})
  \leq\ED(\rho_{AB})
  \leq\ER(\rho_{AB})
  \leq\EF(\rho_{AB}),
  \qquad
  \ER(\rho_{AB})\leq\IAB ,
  \label{eq:measure-hierarchy}
\end{equation}
Here, the entanglement of formation $E_F$, which measures the amount of entanglement needed to prepare the state, is
\begin{equation}
  \EF(\rho_{AB})
  \defeq
  \inf_{\rho_{AB}=\sum_i p_i\ket{\psi_i}\!\bra{\psi_i}}
  \sum_i p_i
  S\!\left(\Tr_B\ket{\psi_i}\!\bra{\psi_i}\right)
  \label{eq:ef-convex-roof}
\end{equation}
and $I(A:B) \defeq S(\rho_A) + S(\rho_B) - S(\rho_{AB})$ denotes the mutual information. 
The distillable entanglement $E_D(\rho_{AB})$ is the optimal asymptotic rate at which maximally entangled states across $A:B$ can be extracted by local operations and classical communication (LOCC). 
The hashing bound is given by
\begin{equation}
  \mathrm{hash}(\rho_{AB})
  \defeq\max\bigl\{0,\,
    S(\rho_A)-S(\rho_{AB}),\,
    S(\rho_B)-S(\rho_{AB})\bigr\}.
  \label{eq:hashing-bound}
\end{equation}
For our purpose, the upper bound will be crucial:
\begin{equation}
  \ER(\rho_{AB})
  \leq\min\bigl\{\EF(\rho_{AB}),\,\IAB\bigr\}. \label{eq:upper_bound}
\end{equation}

\subsection{Main result}

Let $\ket{\psi}_{ABC}$ be drawn from the Haar measure on the unit sphere of
$\mathbb{C}^{d_A}\otimes\mathbb{C}^{d_B}\otimes\mathbb{C}^{d_C}$. Let
\begin{equation}
  \rho_{AB}= \Tr_C\ket{\psi}\!\bra{\psi}_{ABC}
  \label{eq:rho}
\end{equation}
be the reduced state on the bipartition $A{:}B$.
Throughout this paper, we write
\begin{equation}
  \dmin\defeq\min\{d_A,d_B\},\qquad
  \dmax\defeq\max\{d_A,d_B\},\qquad
  D\defeq d_Ad_B=\dmin\dmax. 
  \label{eq:dims}
\end{equation}
The regime of interest in this paper is
\begin{equation}
  d_C\leq D,
  \label{eq:regime}
\end{equation}
since, for $d_C \geq D$, the reduced state $\rho_{AB}$ is close to the maximally mixed state, and hence $E_{R}$ vanishes. 
It will be convenient to define 
\begin{equation}
    \tilde{d} \,\defeq \frac{D}{\max\{d_C,\dmax\}}
    = 
    \begin{cases}
        d_Ad_B/d_C, \qquad &d_C \geq d_+\\
        d_-, \qquad &d_C \leq d_+ 
    \end{cases}
    .
\end{equation}

Our primary goal is to compute the relative entropy of entanglement $E_{R}(\rho_{AB})$ for the reduced state $\rho_{AB}$ derived from a tripartite Haar random state $\ket{\psi}_{ABC}$.

\begin{theorem}[Main result]
\label{thm:main}
Let $\rho_{AB}$ be the reduced state of a tripartite Haar random state $\ket{\psi}_{ABC}$ as in Eq.~\eqref{eq:rho} with $1\leq d_C\leq D$.
There exist absolute constants $K,C,c>0$ such that
\begin{equation}
  \mathbb{P}\!\left[
    \abs{\ER(\rho_{AB})-\log \tilde{d} \ }>K
  \right]
  \leq C\exp\!\left[-c(d_A+d_B+d_C)\right].
  \label{eq:main-tail}
\end{equation}
\end{theorem}

In other words, we essentially find $E_{R} = \log \tilde{d} + O(1)$ where $O(1)$ denotes an absolute additive constant, independent of $d_A$, $d_B$ and $d_C$.
This statement holds \emph{with high probability} over the Haar measure. 
For Haar random states, the following results are well known~\cite{Hayden:2005sqo,Page:1993df,PhysRevLett.77.1}:
\begin{align}
E_{F}(\rho_{AB}) = \log d_- + O(1), \qquad I(A:B) = \log\frac{d_Ad_B}{d_C} + O(1).
\end{align}
Hence, we find that $E_{R}$ nearly saturates the upper bound Eq.~\eqref{eq:upper_bound}. 
 
\section{Bound on upper tail: one-sided classicalization}

We first bound the upper tails of the relative entropy of entanglement. 
We will show that, for some absolute constants $K,C,c>0$, 
\begin{equation}
  \mathbb{P}\!\left[
    \ER(\rho_{AB}) > \log \tilde{d} + K
  \right]
  \leq C\exp\!\left[-c(d_Ad_B+d_Ad_C+d_Bd_C)\right].
  \label{eq:upper-tail}
\end{equation}
This is stronger than the probability bound required in Theorem~\ref{thm:main}.

Our basic strategy is to construct an explicit separable state $\sigma_{AB}$ such that
\begin{equation}
D(\rho_{AB}\Vert\sigma_{AB}) = \log\tilde{d} +O(1)
\end{equation}
with high probability. 
Since $E_{R}(\rho_{AB})\leq D(\rho_{AB}\Vert\sigma_{AB})$, this immediately gives the desired upper bound.

Exchanging $A$ and $B$ if necessary, assume that
$d_A=d_-$ and $d_B=d_+$. Let the Schmidt decomposition of $\ket{\psi}$ across $A:BC$ be
\begin{equation}
  \ket{\psi}_{ABC}
  =
  \sum_{k=1}^{r_A}
  \sqrt{q_k}\,
  \ket{u_k}_A \ket{v_k}_{BC},
  \qquad
  \beta_k^B
  \defeq
  \Tr_C \ket{v_k}\!\bra{v_k}.
  \label{eq:schmidt-classicalization}
\end{equation}
Here $r_A=\rank\rho_A$. Dephasing $A$ in the Schmidt basis
$\{\ket{u_k}\}$ and then tracing out $C$ gives the separable state
\begin{equation}
  \sigma^{A\to\mathrm{cl}}_{AB}
  \,\defeq\,
  \sum_{k=1}^{r_A}
  q_k
  \ket{u_k}\!\bra{u_k}_A
  \otimes
  \beta_k^B,
  \label{eq:sigma-classicalized}
\end{equation}
which we call the one-sided classicalization of $\rho_{AB}$.

\begin{lemma}[One-sided classicalization]
\label{lem:one-sided}
The support of $\rho_{AB}$ is contained in that of one-sided classicalization $\sigma^{A\to\mathrm{cl}}_{AB}$, and
\begin{equation}
  D\!\left(
    \rho_{AB}
    \middle\Vert
    \sigma^{A\to\mathrm{cl}}_{AB}
  \right)
  =
  -S(\rho_{AB})
  +S(\rho_A)
  +\sum_k q_k S(\beta_k).
  \label{eq:exact-dcl}
\end{equation}
\end{lemma}

Note that Lemma~\ref{lem:one-sided} is completely general and does not rely on the Haar randomness. 

\begin{proof}
Choose an orthonormal basis $\{\ket{\alpha}_C\}$ and expand each Schmidt
vector as
\begin{equation}
  \ket{v_k}_{BC}
  =
  \sum_\alpha
  \ket{\xi_{k\alpha}}_B\ket{\alpha}_C,
  \qquad
  \beta_k
  =
  \sum_\alpha
  \ket{\xi_{k\alpha}}\!\bra{\xi_{k\alpha}}
  \label{eq:vk-components}
\end{equation}
with unnormalized vectors $\ket{\xi_{k\alpha}}$.
Tracing Eq.~\eqref{eq:schmidt-classicalization} over $C$ gives
\begin{equation}
  \rho_{AB}
  =
  \sum_{k,\ell,\alpha}
  \sqrt{q_kq_\ell}\,
  \ket{u_k}\!\bra{u_\ell}_A
  \otimes
  \ket{\xi_{k\alpha}}\!\bra{\xi_{\ell\alpha}}_B.
  \label{eq:rho-classicalization-blocks}
\end{equation}
From this expression, every vector in the support of $\rho_{AB}$ belongs
to the direct sum
\begin{equation}
  \bigoplus_k
  \ket{u_k}_A\otimes\supp\beta_k
  =
  \supp\sigma^{A\to\mathrm{cl}}_{AB},
\end{equation}
proving the support containment.

Since $\sigma^{A\to\mathrm{cl}}_{AB}$ is block diagonal with respect to
these orthogonal Schmidt sectors,
\begin{equation}
  \log\sigma^{A\to\mathrm{cl}}_{AB}
  =
  \bigoplus_k
  \ket{u_k}\!\bra{u_k}_A
  \otimes
  \bigl[(\log q_k) \mathbf{1}+\log\beta_k\bigr]_B.
  \label{eq:log-sigma-classicalized}
\end{equation}
Consequently, only the diagonal blocks ($k=\ell$) of
Eq.~\eqref{eq:rho-classicalization-blocks} contribute to
$\Tr(\rho_{AB}\log\sigma^{A\to\mathrm{cl}}_{AB})$. Since the $k$th
diagonal block of $\rho_{AB}$ is $q_k\beta_k$, we obtain
\begin{align}
  \Tr\!\left(
    \rho_{AB}\log\sigma^{A\to\mathrm{cl}}_{AB}
  \right)
  &=
  \sum_k q_k\log q_k
  +
  \sum_k q_k\Tr(\beta_k\log\beta_k)
  \nonumber\\
  &=
  -S(\rho_A)-\sum_k q_kS(\beta_k).
\end{align}
Combining this with
$\Tr(\rho_{AB}\log\rho_{AB})=-S(\rho_{AB})$ proves Eq.~\eqref{eq:exact-dcl}.
\end{proof}

For Haar random states, since
\begin{equation}
  S(\rho_A)\leq\log d_A,
  \qquad
  S(\beta_k)
  \leq
  \log\rank\beta_k
  \leq
  \log\min\{d_B,d_C\},
\end{equation}
Lemma~\ref{lem:one-sided} gives
\begin{equation}
  D\!\left(
    \rho_{AB}
    \middle\Vert
    \sigma^{A\to\mathrm{cl}}_{AB}
  \right)
  \leq
  -S(\rho_{AB})
  +\log d_A
  +\log\min\{d_B,d_C\}.
  \label{eq:dcl-dimensional-bound}
\end{equation}
A standard result for Haar random state is that $S(\rho_{AB}) = \log d_C + O(1)$ for $d_C \leq D$. 
Specifically, the entropy-concentration bound of Hayden, Leung, and Winter~\cite{Hayden:2005sqo} implies that, after appropriately adjusting the absolute constants and noting that $d_A^{-1}+d_B^{-1}+d_C^{-1}\lesssim (\log d_C)^{-2}$,
\begin{equation}
  \mathbb{P}\!\left[
    S(\rho_{AB})<\log d_C-K
  \right]
  \leq
  C\exp\!\left[
    -c\bigl(d_A d_B+d_A d_C+d_B d_C\bigr)
  \right].
  \label{eq:entropy-lower-tail}
\end{equation}
Together with Eq.~\eqref{eq:dcl-dimensional-bound}
this proves
Eq.~\eqref{eq:upper-tail}.

\section{Bound on lower tail: overlap with product states}

We next bound the lower tail of the relative entropy of entanglement. We will show that, for some absolute constants $K,C,c>0$,
\begin{equation}
  \mathbb{P}\!\left[
    \ER(\rho_{AB})<\log\widetilde d-K
  \right]
  \leq
  C\exp\!\left[-c(d_A+d_B+d_C)\right].
  \label{eq:lower-tail}
\end{equation}
Together with the upper-tail bound in Eq.~\eqref{eq:upper-tail}, this completes the proof of Theorem~\ref{thm:main}.

Our strategy is to lower-bound $\ER(\rho_{AB})$ using its maximal overlap with a product state. Define
\begin{equation}
  \Lambda(\rho_{AB})
  \defeq
  \max_{\substack{\|a\|=\|b\|=1}}
  \bra{a,b}\rho_{AB}\ket{a,b}.
  \label{eq:maximal-product-overlap}
\end{equation}
Since $\rho_{AB}=\Tr_C\ket{\psi}\!\bra{\psi}$, this can equivalently be written as the maximal overlap for the purified state $|\psi\rangle_{ABC}$:
\begin{equation}
  \Lambda(\rho_{AB})
  =
  \max_{\substack{\|a\|=\|b\|=\|c\|=1}}
  \left|\langle a,b,c|\psi\rangle\right|^2.
  \label{eq:purified-product-overlap}
\end{equation}


\begin{lemma}[Product-overlap bound]
\label{lem:product-overlap-bound}
For any bipartite state $\rho_{AB}$,
\begin{equation}
  \ER(\rho_{AB})
  \geq
  -S(\rho_{AB})-\log\Lambda(\rho_{AB}).
  \label{eq:ER-product-overlap-bound}
\end{equation}
\end{lemma}

\begin{proof}
Let $\sigma_{AB}$ be an arbitrary separable state. If
$\supp\rho_{AB}\nsubseteq\supp\sigma_{AB}$, then
$D(\rho_{AB}\Vert\sigma_{AB})=+\infty$, and the claim is immediate.
Otherwise, concavity of the logarithm gives
\begin{equation}
  \Tr\!\left(\rho_{AB}\log\sigma_{AB}\right)
  \leq
  \log\Tr\!\left(\rho_{AB}\sigma_{AB}\right).
  \label{eq:log-jensen}
\end{equation}
Writing
\begin{equation}
  \sigma_{AB}
  =
  \sum_i p_i
  \ket{a_i,b_i}\!\bra{a_i,b_i},
\end{equation}
we have
\begin{equation}
  \Tr\!\left(\rho_{AB}\sigma_{AB}\right)
  =
  \sum_i p_i
  \bra{a_i,b_i}\rho_{AB}\ket{a_i,b_i}
  \leq
  \Lambda(\rho_{AB}).
\end{equation}
Therefore,
\begin{align}
  D(\rho_{AB}\Vert\sigma_{AB})
  &=
  -S(\rho_{AB})
  -
  \Tr\!\left(\rho_{AB}\log\sigma_{AB}\right)
  \nonumber\\
  &\geq
  -S(\rho_{AB})
  -
  \log\Lambda(\rho_{AB}).
\end{align}
Taking the infimum over all separable $\sigma_{AB}$ proves the claim.
\end{proof}

Thus, it remains to upper-bound the maximal product overlap $\Lambda(\rho_{AB})$.

Bounds of this type are known in the mathematical literature~\cite{TomiokaSuzuki}, typically formulated for Gaussian tensors.
For completeness, we state the result directly in quantum-mechanical language and give a short proof using an epsilon net. 
Our bound also improves upon a previous result in the quantum information  literature~\cite{ZhuChenHayashi} by removing the logarithmic factor. 

\begin{lemma}[Maximal product overlap]
\label{lem:haar-product-overlap}
Let $\ket{\psi}_{ABC}$ be Haar random and let
$\rho_{AB}=\Tr_C\ket{\psi}\!\bra{\psi}$. 
There exist absolute constants
$C,c>0$ such that
\begin{equation}
  \mathbb{P}\!\left[
    \Lambda(\rho_{AB})
    >
    C\frac{d_A+d_B+d_C}{d_A d_B d_C}
  \right]
  \leq
  C\exp\!\left[-c(d_A+d_B+d_C)\right].
  \label{eq:haar-product-overlap}
\end{equation}
\end{lemma}


\begin{proof}
For any fixed product state $\ket{a,b,c}$, its overlap with a Haar-random state obeys the sub-Gaussian tail bound 
\begin{equation}
  \mathbb{P}\!\left[
    \left|\langle a,b,c|\psi\rangle\right|
    >
    \frac{t}{\sqrt{d_A d_Bd_C}}
  \right]
  \leq
  e^{-c_0 t^2}
  \label{eq:fixed-product-tail}
\end{equation}
for some $c_0>0$.

An $\epsilon$-net is a finite set of states such that every state lies within distance $\epsilon$ of some state in the set.
The unit sphere in $\mathbb{C}^{d_X}$ admits such a net containing at most~$\left(\frac{C_0}{\epsilon}\right)^{2d_X}$ states for some $C_0>0$~\cite{Hayden:2004mfu}. Thus, for any fixed constant $\epsilon>0$, the set of product
states admits an epsilon net of size
\begin{equation}
  \exp\!\left[C_1(d_A+d_B+d_C)\right]
\end{equation}
for some $C_1>0$.
Let $\mathcal{N}_X$ denote the local epsilon nets. 
Define the maximal overlap on the product net by
\begin{equation}
  \Lambda_{\mathrm{net}}(\rho_{AB})
  \defeq
  \max_{\substack{
    a\in\mathcal{N}_A,
    b\in\mathcal{N}_B,
    c\in\mathcal{N}_C
  }}
  |\langle a,b,c|\psi\rangle|^2.
\end{equation}
Applying Eq.~\eqref{eq:fixed-product-tail} to every state in the product
net and taking a union bound gives
\begin{equation}
  \mathbb{P}\!\left[
    \Lambda_{\mathrm{net}}(\rho_{AB})
    >
    \frac{t^2}{d_A d_B d_C}
  \right]
  \leq
  \exp\!\left[
    C_1(d_A+d_B+d_C)-c_0t^2
  \right].
\end{equation}
Choosing
\begin{equation}
  t=C_2\sqrt{d_A+d_B+d_C},
\end{equation}
with $C_2$ sufficiently large, we obtain
\begin{equation}
  \Lambda_{\mathrm{net}}(\rho_{AB})
  \leq
  C_2^2
  \frac{d_A+d_B+d_C}{d_A d_B d_C}
  \label{eq:net-overlap-bound}
\end{equation}
except with probability at most
$C\exp\!\left[-c(d_A+d_B+d_C)\right]$.

It remains to compare the maximum over the net with the maximum over all
product states. Choose unit vectors $a,b,c$ attaining the maximum
$\Lambda(\rho_{AB})$, and choose
$a_0\in\mathcal{N}_A$, $b_0\in\mathcal{N}_B$, and
$c_0\in\mathcal{N}_C$ such that
\begin{equation}
  \|a-a_0\|,
  \|b-b_0\|,
  \|c-c_0\|
  \leq
  \epsilon.
\end{equation}
Writing
\begin{equation}
  a=a_0+\delta a,
  \qquad
  b=b_0+\delta b,
  \qquad
  c=c_0+\delta c,
\end{equation}
and expanding the trilinear overlap, the three first-order, three
second-order, and one third-order error terms are bounded in total by
\begin{equation}
  \left(3\epsilon+3\epsilon^2+\epsilon^3\right)
  \sqrt{\Lambda(\rho_{AB})}.
\end{equation}
Therefore,
\begin{equation}
  \sqrt{\Lambda(\rho_{AB})}
  \leq
  \sqrt{\Lambda_{\mathrm{net}}(\rho_{AB})}
  +
  \left(3\epsilon+3\epsilon^2+\epsilon^3\right)
  \sqrt{\Lambda(\rho_{AB})}.
\end{equation}
Taking $\epsilon=1/8$, we have $3\epsilon+3\epsilon^2+\epsilon^3  =
  \frac{217}{512}  <
  \frac{1}{2}$, and hence
\begin{equation}
  \Lambda(\rho_{AB})
  \leq
  4\Lambda_{\mathrm{net}}(\rho_{AB}).
\end{equation}
Combining this with Eq.~\eqref{eq:net-overlap-bound} and absorbing the
constant factors into $C$, we conclude that
\begin{equation}
  \Lambda(\rho_{AB})
  \leq
  C\frac{d_A+d_B+d_C}{d_A d_B d_C}
\end{equation}
except with probability at most $C\exp\!\left[-c(d_A+d_B+d_C)\right]$.
This proves the claim.
\end{proof}

We now combine Lemmas~\ref{lem:product-overlap-bound} and
\ref{lem:haar-product-overlap}. 
Since $\rank\rho_{AB}\leq d_C$, we have $S(\rho_{AB})\leq\log d_C$. 
Therefore, outside an event of probability at most $C\exp[-c(d_A+d_B+d_C)]$,
\begin{align}
  \ER(\rho_{AB})
  &\geq
  -S(\rho_{AB})-\log\Lambda(\rho_{AB})
  \nonumber\\
  &\geq
  \log\frac{d_A d_B}{d_A+d_B+d_C}
  -O(1) \nonumber \\
  &= \log\tilde d-O(1).\nonumber
\end{align}
This proves Eq.~\eqref{eq:lower-tail} and completes the proof of
Theorem~\ref{thm:main}.

\section{Discussion}
\label{sec:discussion}

We have determined the relative entropy of entanglement of a Haar-induced mixed state up to an absolute additive constant, with the one-sided classicalized state $\sigma^{A\to\mathrm{cl}}_{AB}$ being nearly the closest.
Because $\sigma^{A\to\mathrm{cl}}_{AB}$ is obtained by dephasing $A$ in the Schmidt basis, its relative entropy also admits the interpretation
\begin{equation}
  D\!\left(
    \rho_{AB}
    \middle\Vert
    \sigma^{A\to\mathrm{cl}}_{AB}
  \right)
  =
  S\!\left(\sigma^{A\to\mathrm{cl}}_{AB}\right)
  -S(\rho_{AB})
\end{equation}
as an entropy drop due to Schmidt dephasing.
Equivalently, since the dephasing preserves both marginals, this is the mutual information lost by measuring the Schmidt label of $A$ and discarding the measurement outcome, also known as the quantum discord~\cite{Henderson:2001wrr,PhysRevLett.88.017901,Zurek:2011vew,Mori:2025gqe}.

Our result also reveals a macroscopic separation between the relative entropy of entanglement and distillable entanglement. In the regime
$d_A,d_B\ll d_C$, the one-shot, one-way distillable entanglement vanishes asymptotically~\cite{ Hayden:2005sqo,Li:2025nxv}, whereas $\ER(\rho_{AB})  =
  \log\frac{d_A d_B}{d_C}
  +O(1)$ 
remains macroscopically large. More generally, the one-way distillable entanglement saturates the hashing bound~\cite{Mori:2024gwe}, while
Theorem~\ref{thm:main} shows that $\ER(\rho_{AB})$ can exceed this bound by a macroscopic amount. Haar-random states therefore contain extensive entanglement that is inaccessible to these distillation protocols.

We also note that, as a corollary to our results, we obtain the large-$d$ asymptotic formula for the geometric measure of tripartite entanglement~\cite{Wei:2003qfk,Weinbrenner:2025uwb}:
\begin{equation}
    G(\psi)\defeq -\log \max_{\substack{\|a\|=\|b\|=\|c\|=1}}
  \left|\langle a,b,c|\psi\rangle\right|^2 = -\log \Lambda(\rho_{AB})= \log (d_- \, d_C) +O(1),\quad d_C\le D
\end{equation}
as implied in an earlier mathematical literature~\cite{Dartois:2024zuc}.
This is related to the maximum number of states that can be discriminated under LOCC~\cite{Hayashi:2006ppb}.

Several questions remain open. Most directly, our theorem concerns the single-copy quantity $\ER$, while the generalized quantum Stein's lemma~\cite{Brandao:2010iez} involves its regularized version
$
  \ER^\infty(\rho)
  \defeq
  \lim_{n\to\infty}
  \frac{1}{n}\ER\!\left(\rho^{\otimes n}\right)$.
Operationally, $\ER^\infty$ determines the optimal asymptotic error exponent for distinguishing many copies of $\rho$ from the set of separable states. 
It would therefore be interesting to determine whether regularization changes the leading Haar-random answer. 

Another question is how broadly the one-sided classicalization remains near optimal. 
Writing the reduced state as $\rho_A\propto e^{-K_A}$, where $K_A$ is the modular/entanglement Hamiltonian, the Schmidt basis is precisely the eigenbasis of $K_A$. 
One-sided classicalization can therefore be viewed as dephasing $\rho_{AB}$ in the modular-energy basis of subsystem $A$. 
The same construction gives a general upper bound on $\ER$ for any tripartite pure state, but there is no reason for it to be tight in a structured many-body system.
In this setting, Schmidt dephasing may be viewed as dephasing in the eigenbasis of the modular Hamiltonian. Understanding the resulting quantity in quantum many-body systems, conformal field theories, and random tensor networks may be interesting future problems. 

The original motivation for this work came from the possible holographic interpretation of the relative entropy of entanglement in the AdS/CFT correspondence. 
In holographic states, we find that $\ER$ can deviate from both $\EF$ and the mutual information and is instead described by a new minimal-surface prescription.

\subsection*{Acknowledgment}
We thank Satoya Imai and Ryuji Takagi for helpful comments made during the YITP workshop YITP-W-26-07, held at the Yukawa Institute for Theoretical Physics, Kyoto University.
Research at Perimeter Institute is supported in part by the Government of Canada through the Department of Innovation, Science and Economic Development and by the Province of Ontario through the Ministry of Colleges and Universities. 
This work was supported by the RIKEN TRIP initiative, JSPS KAKENHI Grant Number 23KJ1154, 24K17047, and the Applied Quantum Computing Challenge Program at the National Research Council of Canada.

\mciteSetMidEndSepPunct{}{\ifmciteBstWouldAddEndPunct.\else\fi}{\relax}
\bibliographystyle{JHEP}
\bibliography{ref.bib}

@article{Hayden:2005sqo,
    author = "Hayden, Patrick and Leung, Debbie W. and Winter, Andreas",
    title = "{Aspects of Generic Entanglement}",
    eprint = "quant-ph/0407049",
    archivePrefix = "arXiv",
    doi = "10.1007/s00220-006-1535-6",
    journal = "Commun. Math. Phys.",
    volume = "265",
    number = "1",
    pages = "95--117",
    year = "2006"
}

@article{Li:2025nxv,
    author = "Li, Zhi and Mori, Takato and Yoshida, Beni",
    title = "{Tripartite Haar random state has no bipartite entanglement}",
    eprint = "2502.04437",
    archivePrefix = "arXiv",
    primaryClass = "quant-ph",
    reportNumber = "YITP-25-15, RUP-25-8",
    month = "2",
    year = "2025"
}

@article{Mori:2024gwe,
    author = "Mori, Takato and Yoshida, Beni",
    title = "{Does connected wedge imply distillable entanglement?}",
    eprint = "2411.03426",
    archivePrefix = "arXiv",
    primaryClass = "hep-th",
    reportNumber = "YITP-24-149",
    doi = "10.1007/JHEP01(2026)125",
    journal = "JHEP",
    volume = "01",
    pages = "125",
    year = "2026"
}

@article{TomiokaSuzuki,
  author        = {Ryota Tomioka and Taiji Suzuki},
  title         = {Spectral norm of random tensors},
  eprint        = {1407.1870},
  archivePrefix = {arXiv},
  primaryClass  = {math.ST},
  year          = {2014}
}

@article{ZhuChenHayashi,
  author        = {Zhu, Huangjun and Chen, Lin and Hayashi, Masahito},
  title         = {Additivity and non-additivity of multipartite entanglement measures},
  journal       = {New Journal of Physics},
  volume        = {12},
  pages         = {083002},
  year          = {2010},
  doi           = {10.1088/1367-2630/12/8/083002},
  eprint        = {1002.2511},
  archivePrefix = {arXiv},
  primaryClass  = {quant-ph}
}

@article{Aubrun:2014threshold,
    author = "Aubrun, Guillaume and Szarek, Stanislaw J. and Ye, Deping",
    title = "{Entanglement Thresholds for Random Induced States}",
    eprint = "1106.2264",
    archivePrefix = "arXiv",
    primaryClass = "quant-ph",
    doi = "10.1002/cpa.21460",
    journal = "Commun. Pure Appl. Math.",
    volume = "67",
    number = "1",
    pages = "129--171",
    year = "2014"
}

@article{Devetak:2003zfw,
    author = "Devetak, Igor and Winter, Andreas",
    title = "{Distillation of Secret Key and Entanglement from Quantum States}",
    eprint = "quant-ph/0306078",
    archivePrefix = "arXiv",
    doi = "10.1098/rspa.2004.1372",
    journal = "Proc. Roy. Soc. Lond. A",
    volume = "461",
    pages = "207--235",
    year = "2005"
}

@article{Horodecki:2009zz,
    author = "Horodecki, Ryszard and Horodecki, Pawel and Horodecki, Michal and
    Horodecki, Karol",
    title = "{Quantum Entanglement}",
    eprint = "quant-ph/0702225",
    archivePrefix = "arXiv",
    doi = "10.1103/RevModPhys.81.865",
    journal = "Rev. Mod. Phys.",
    volume = "81",
    pages = "865--942",
    year = "2009"
}

@article{Huang_2014,
doi = {10.1088/1367-2630/16/3/033027},
year = {2014},
month = {mar},
publisher = {IOP Publishing},
volume = {16},
number = {3},
pages = {033027},
author = {Huang, Yichen},
title = {Computing quantum discord is NP-complete},
journal = {New Journal of Physics},
}

@article{Lu:2020jza,
    author = "Lu, Tsung-Cheng and Grover, Tarun",
    title = "{Entanglement transitions as a probe of quasiparticles and quantum thermalization}",
    eprint = "2008.11727",
    archivePrefix = "arXiv",
    primaryClass = "cond-mat.stat-mech",
    doi = "10.1103/PhysRevB.102.235110",
    journal = "Phys. Rev. B",
    volume = "102",
    number = "23",
    pages = "235110",
    year = "2020"
}

@article{Shapourian:2020mkc,
    author = "Shapourian, Hassan and Liu, Shang and Kudler-Flam, Jonah and Vishwanath, Ashvin",
    title = "{Entanglement Negativity Spectrum of Random Mixed States: A Diagrammatic Approach}",
    eprint = "2011.01277",
    archivePrefix = "arXiv",
    primaryClass = "cond-mat.str-el",
    doi = "10.1103/PRXQuantum.2.030347",
    journal = "PRXQuantum",
    volume = "2",
    number = "3",
    pages = "030347",
    year = "2021"
}

@article{Vedral:1997hd,
    author = "Vedral, V. and Plenio, M. B.",
    title = "{Entanglement measures and purification procedures}",
    eprint = "quant-ph/9707035",
    archivePrefix = "arXiv",
    doi = "10.1103/PhysRevA.57.1619",
    journal = "Phys. Rev. A",
    volume = "57",
    pages = "1619--1633",
    year = "1998"
}

@article{Vedral:1997qn,
    author = "Vedral, V. and Plenio, M. B. and Rippin, M. A. and Knight, P. L.",
    title = "{Quantifying entanglement}",
    eprint = "quant-ph/9702027",
    archivePrefix = "arXiv",
    doi = "10.1103/PhysRevLett.78.2275",
    journal = "Phys. Rev. Lett.",
    volume = "78",
    pages = "2275--2279",
    year = "1997"
}

@article{Chen:2002ubp,
    author = "Chen, Yi-Xin and Yang, Dong",
    title = "{The Relative Entropy of Entanglement of Schmidt Correlated States and Distillation}",
    eprint = "quant-ph/0204152",
    archivePrefix = "arXiv",
    doi = "10.1023/A:1023469830754",
    journal = "Quant. Inf. Proc.",
    volume = "1",
    number = "5",
    pages = "389--395",
    year = "2002"
}

@article{Audenaert:2002lkg,
    author = "Audenaert, K. and De Moor, B. and Vollbrecht, K. G. H. and Werner, R. F.",
    title = "{Asymptotic Relative Entropy of Entanglement for Orthogonally Invariant States}",
    eprint = "quant-ph/0204143",
    archivePrefix = "arXiv",
    doi = "10.1103/PhysRevA.66.032310",
    journal = "Phys. Rev. A",
    volume = "66",
    pages = "032310",
    year = "2002"
}

@article{Vollbrecht:2001vsi,
    author = "Vollbrecht, K. G. H. and Werner, R. F.",
    title = "{Entanglement measures under symmetry}",
    eprint = "quant-ph/0010095",
    archivePrefix = "arXiv",
    doi = "10.1103/PhysRevA.64.062307",
    journal = "Phys. Rev. A",
    volume = "64",
    number = "6",
    pages = "062307",
    year = "2001"
}

@article{Kim:2010fps,
    author = "Kim, Hungsoo and Hwang, Mi-Ra and Jung, Eylee and Park, DaeKil",
    title = "{Difficulties in analytic computation for relative entropy of entanglement}",
    eprint = "1002.4695",
    archivePrefix = "arXiv",
    primaryClass = "quant-ph",
    doi = "10.1103/PhysRevA.81.052325",
    month = "4",
    year = "2010"
}

@article{Hajdusek:2012eqs,
    author = "Hajdu{\v{s}}ek, Michal and Murao, Mio",
    title = "{Direct evaluation of pure graph state entanglement}",
    eprint = "1207.5877",
    archivePrefix = "arXiv",
    primaryClass = "quant-ph",
    doi = "10.1088/1367-2630/15/1/013039",
    journal = "New J. Phys.",
    volume = "15",
    pages = "013039",
    year = "2013"
}

@article{Markham:2007vcs,
    author = "Markham, Damian and Miyake, Akimasa and Virmani, Shashank",
    title = "{Entanglement and local information access for graph states}",
    eprint = "quant-ph/0609102",
    archivePrefix = "arXiv",
    doi = "10.1088/1367-2630/9/6/194",
    journal = "New J. Phys.",
    volume = "9",
    pages = "194",
    year = "2007"
}

@article{Page:1993df,
    author = "Page, Don N.",
    title = "{Average entropy of a subsystem}",
    eprint = "gr-qc/9305007",
    archivePrefix = "arXiv",
    reportNumber = "ALBERTA-THY-22-93",
    doi = "10.1103/PhysRevLett.71.1291",
    journal = "Phys. Rev. Lett.",
    volume = "71",
    pages = "1291--1294",
    year = "1993"
}

@article{PhysRevLett.77.1,
  title = {Average Entropy of a Quantum Subsystem},
  author = {Sen, Siddhartha},
  journal = {Phys. Rev. Lett.},
  volume = {77},
  issue = {1},
  pages = {1--3},
  numpages = {0},
  year = {1996},
  month = {Jul},
  publisher = {American Physical Society},
  doi = {10.1103/PhysRevLett.77.1},
  url = {https://link.aps.org/doi/10.1103/PhysRevLett.77.1}
}

@article{Hayden:2004mfu,
    author = "Hayden, Patrick and Leung, Debbie and Shor, Peter W. and Winter, Andreas",
    title = "{Randomizing Quantum States: Constructions and Applications}",
    eprint = "quant-ph/0307104",
    archivePrefix = "arXiv",
    doi = "10.1007/s00220-004-1087-6",
    journal = "Commun. Math. Phys.",
    volume = "250",
    number = "2",
    pages = "371--391",
    year = "2004"
}

@article{Mori:2025gqe,
    author = "Mori, Takato",
    title = "{Quantum correlation beyond entanglement: Holographic discord and multipartite generalizations}",
    eprint = "2506.02131",
    archivePrefix = "arXiv",
    primaryClass = "hep-th",
    reportNumber = "RUP-25-11, YITP-25-66",
    month = "6",
    year = "2025"
}

@article{Henderson:2001wrr,
    author = "Henderson, L. and Vedral, V.",
    title = "{Classical, quantum and total correlations}",
    eprint = "quant-ph/0105028",
    archivePrefix = "arXiv",
    doi = "10.1088/0305-4470/34/35/315",
    journal = "J. Phys. A",
    volume = "34",
    number = "35",
    pages = "6899",
    year = "2001"
}

@article{PhysRevLett.88.017901,
  title = {Quantum Discord: A Measure of the Quantumness of Correlations},
  author = {Ollivier, Harold and Zurek, Wojciech H.},
  journal = {Phys. Rev. Lett.},
  volume = {88},
  issue = {1},
  pages = {017901},
  numpages = {4},
  year = {2001},
  month = {Dec},
  publisher = {American Physical Society},
  doi = {10.1103/PhysRevLett.88.017901},
  url = {https://link.aps.org/doi/10.1103/PhysRevLett.88.017901}
}

@inproceedings{Zurek:2011vew,
    author = "Zurek, W. H.",
    title = "{Einselection and Decoherence from an Information Theory Perspective}",
    eprint = "quant-ph/0011039",
    archivePrefix = "arXiv",
    doi = "10.1002/1521-3889(200011)9:11/12<855::AID-ANDP855>3.0.CO;2-K",
    month = "6",
    year = "2011"
}

@article{Brandao:2010iez,
    author = "Brandao, Fernando G. S. L. and Plenio, Martin B.",
    title = "{A Generalization of Quantum Stein's Lemma}",
    eprint = "0904.0281",
    archivePrefix = "arXiv",
    primaryClass = "quant-ph",
    doi = "10.1007/s00220-010-1005-z",
    journal = "Commun. Math. Phys.",
    volume = "295",
    pages = "791",
    year = "2010"
}

@article{Wei:2003qfk,
    author = "Wei, Tzu-Chieh and Goldbart, Paul M.",
    title = "{Geometric measure of entanglement and applications to bipartite and multipartite quantum states}",
    eprint = "quant-ph/0307219",
    archivePrefix = "arXiv",
    doi = "10.1103/PhysRevA.68.042307",
    journal = "Phys. Rev. A",
    volume = "68",
    number = "4",
    pages = "042307",
    year = "2003"
}

@article{Dartois:2024zuc,
    author = "Dartois, Stephane and McKenna, Benjamin",
    title = "{Injective norm of real and complex random tensors I: From spin glasses to geometric entanglement}",
    eprint = "2404.03627",
    archivePrefix = "arXiv",
    primaryClass = "math.PR",
    month = "4",
    year = "2024"
}

@article{Weinbrenner:2025uwb,
    author = {Weinbrenner, Lisa T. and G{\"u}hne, Otfried},
    title = "{Quantifying entanglement from the geometric perspective}",
    eprint = "2505.01394",
    archivePrefix = "arXiv",
    primaryClass = "quant-ph",
    doi = "10.1209/0295-5075/adffb5",
    journal = "EPL",
    volume = "151",
    number = "6",
    pages = "68001",
    year = "2025"
}

@article{Hayashi:2006ppb,
    author = "Hayashi, M. and Markham, D. and Murao, M. and Owari, M. and Virmani, S.",
    title = "{Bounds on Multipartite Entangled Orthogonal State Discrimination Using Local Operations and Classical Communication}",
    eprint = "quant-ph/0506170",
    archivePrefix = "arXiv",
    doi = "10.1103/PhysRevLett.96.040501",
    journal = "Phys. Rev. Lett.",
    volume = "96",
    pages = "040501",
    year = "2006"
}
\end{document}